%% file: main.tex
\documentclass[letterpaper,cleveref, autoref, thm-restate]{lipics-v2021}
\keywords{Transaction ordering; First-come-first-serve; First-price auctions.}

\usepackage[utf8]{inputenc}
\usepackage{amsmath}
\usepackage{amssymb}
\usepackage{amsthm}
\usepackage{comment}
\usepackage[dvipsnames]{xcolor}
\usepackage{cite}
\usepackage{xspace}

\newtheorem{property}[theorem]{Property}

\title{Buying Time: Latency Racing vs. Bidding for Transaction Ordering}

\author{Akaki Mamageishvili}{Offchain Labs}{amamageishvili@offchainlabs.com}{}{}
\author{Mahimna Kelkar\footnote{This work was completed in the author's role at Offchain Labs.}}{Cornell University}{mahimna@cs.cornell.edu}{}{}
\author{Jan Christoph Schlegel}{City, University of London}{jchschlegel@gmail.com}{}{}
\author{Edward W. Felten}{Offchain Labs}{ed@offchainlabs.com}{}{}

\Copyright{Akaki Mamageishvili and Mahimna Kelkar and Jan Christoph Schlegel and Edward W. Felten}

\EventEditors{Joseph Bonneau and S. Matthew Weinberg}
\EventNoEds{2}
\EventLongTitle{5th Conference on Advances in Financial Technologies (AFT 2023)}
\EventShortTitle{AFT 2023}
\EventAcronym{AFT}
\EventYear{2023}
\EventDate{October 23-25, 2023}
\EventLocation{Princeton, NJ, USA}
\EventLogo{}
\SeriesVolume{282}
\ArticleNo{23}

\authorrunning{A. Mamageishvili, M. Kelkar, J.C. Schlegel, and E.W. Felten}

\ccsdesc[300]{Theory of computation~Algorithmic game theory}

\begin{document}
\nolinenumbers
\maketitle

\newcommand{\mypara}[1]{\smallskip\noindent\textbf{#1}\hspace*{0.5em}}

\newcommand{\sysname}{\textnormal{\textsf{TimeBoost}}\xspace}

\newcommand{\txset}{\mathcal{T}}
\newcommand{\tx}{\mathsf{tx}}

\newcommand{\Ex}{E}

\input{abstract}

\input{Sections/introduction}
\input{Sections/prelims}
\input{Sections/timeboost}

\input{Sections/economic}
\input{Sections/bidding-comparison}

\input{Sections/conclusion}

\bibliography{refs}

\end{document}

%% file: abstract.tex
\begin{abstract}
    We design \sysname: a practical transaction ordering policy for rollup sequencers that takes into account both transaction timestamps and bids; it works by creating a score from timestamps and bids, and orders transactions based on this score. 
    
    \sysname is transaction-data-independent (i.e., can work with encrypted transactions) and supports low transaction finalization times similar to a first-come first-serve (FCFS or pure-latency) ordering policy. At the same time, it avoids the inefficient latency competition created by an FCFS policy. It further satisfies useful economic properties of first-price auctions that come with a pure-bidding policy. We show through rigorous economic analyses how \sysname allows players to compete on arbitrage opportunities in a way that results in better guarantees compared to both pure-latency and pure-bidding approaches.   
\end{abstract}

%% file: Sections/introduction.tex
\section{Introduction}
\label{sec:introduction}
Transaction ordering is critically important for financial systems---the order in which user transactions are executed can directly impact the profits made by users. This motivates the study of designing transaction ordering policies with useful properties.

In this work, we focus on ordering policies for \textit{centralized sequencers}---meaning that a single \textit{sequencer} receives transactions from users and publishes an ordered sequence to be used for execution. A transaction ordering policy here specifies how the resulting output sequence depends on the contents and arrival times of transactions at the sequencer. Our work provides rigorous economic analyses to justify the utility of our proposed policy. 

\mypara{Why consider a centralized sequencer?}
In addition to the centralized sequencer setting being a potentially simpler model to study as a first step, there are two other main reasons why we choose to do so in this work:

\begin{enumerate}
    \item \underline{\textit{Existing use-cases are already centralized.}}
     Decentralized blockchains such as Ethereum are still \textit{ephemerally centralized} with respect to ordering---for a given block, similar to a centralized sequencer, only a single miner/validator is in complete control of the inclusion and ordering of transactions within the block.     
     Similarly, current layer-2 ``rollup'' protocols (such as Arbitrum and Optimism) also employ a centralized sequencer to order  transactions in a batch posted to the underlying Ethereum base-chain.

    \item \textit{\underline{Ordering policies are mostly orthogonal to the problem of sequencer decentralization}}.
    
    While decentralizing the sequencer is an important active research direction, we note that a suitable transaction ordering policy can be chosen orthogonally to the method of sequencer decentralization. In particular, the decentralized protocol can first be used to agree on single \textit{pre-ordering} or \textit{scoring} of transactions, following which a specific ordering policy can be applied. In other words, the output of the decentralized protocol can be thought of simulating the input of a virtual centralized sequencer on which the ordering policy gets applied. 
    
    \hspace*{1em} An example of this is seen in the recent line of works on fair-ordering ~\cite{kelkar2020order,kelkar2021themis,cachin2022quick,zhang2020oligarchy,kursawe2020wendy}---they can be thought of as a decentralized implementation of a first-come-first-serve ordering policy which combines local transaction orderings from many nodes.
\end{enumerate}

Furthermore, while current centralized sequencer implementations are semi-trusted in that they receive transactions in plaintext and are expected not to deviate from the specified ordering policy or insert transactions of their own, we note that transaction data can be hidden from the sequencer by using threshold decryption by a committee (i.e., the sequencer only sees encrypted transactions and orders them, only after which a committee decrypts the plaintext) or trusted hardware (such as Intel SGX). Through these techniques, the adversarial behavior of the sequencer can be substantially restricted.

The study of ordering policies is important even when the sequencer is trusted (or is suitably constrained as mentioned above) due to the presence of other profit-seeking entities in the system. For instance, after the sequencer publishes state after execution of previous transaction(s), arbitrage opportunities can be created; players in the system will compete with each other to take advantage of these opportunities. Similar situations can also arise due to state updates from external systems. 

\subsection{Existing Ordering Policies}

Ordering policies used on blockchains today fall roughly into three categories described below. 

\mypara{First-come first-serve (FCFS).}
One natural ordering policy is the first-come, first-serve (FCFS) rule. Here, transactions are sequenced in the same order that they were received from users. There are several advantages to FCFS: to begin, it is simple to implement and seems intuitively fair---after all, it is a commonly used policy even for real-world interactions. FCFS also minimizes transaction latency: transactions can be continuously sequenced as they arrive, and do not need to conform to the discrete granularity of blocks. The sequencer in the layer-2 rollup Arbitrum employs an FCFS policy.

One major disadvantage of FCFS however, is that creates \textit{latency competition} in the sense that entities are incentivized to position themselves as close to the sequencer as possible in order to be the first to react to any new market information. This is a well known and studied problem within traditional financial systems. Indeed, high frequency trading (HFT) firms invest millions of dollars into low-latency infrastructure that can operate sub-microsecond or even finer scales; their trading accounts for roughly half of all trading volume~\cite{latency_cost_hft}. This inclination to latency investment is highly inefficient since the investment happens externally to the system (as opposed to bidding; see below) and therefore cannot be used beneficially within the system. Recent works~\cite{str_lat_red, str_peer_select} have also shown the potential for similar strategic manipulation within a pure FCFS protocol in the decentralized setting.

One crucial point to emphasize here is that this latency competition in FCFS \textit{does not disappear even if transaction data is hidden} (e.g., transactions are encrypted). This is because any state changes (from the sequencer or even from external systems) can trigger a profit opportunity wherein it is beneficial to have the quickest access to the sequencer. As a specific example, an update on the trading price of a token can create an arbitrage opportunity whose profit will go only to the player who can submit its transaction to the sequencer first\footnote{Another approach if the sequencer broadcasts state information in a random order to clients is to create many dummy client copies, thereby increasing the chances that some copy gets the feed faster.}. This kind of latency-based arbitrage has already been seen in Arbitrum, which implements a centralized FCFS sequencer. 

\mypara{Per-block transaction bidding.}
A second natural policy is to group transactions into blocks, then order transactions within a block based on their \textit{bid}. Specifically, each transaction is submitted along with a fee or bid; the sequencer now collects all transactions submitted within some time interval and sequences them by the descending order of their bids. This essentially simulates a first-price all-pay auction~\cite{krishna2002auctionbook} (i.e., players bid independently; the highest bid wins but all players need to pay their bid amount) to take advantage of a particular arbitrage opportunity. Since players submit their bids independently, the bidding policy can work as expected even when transactions are encrypted (since state or market updates create arbitrage opportunities).

One advantage of a bidding policy (compared to FCFS) is that the payment is internal to the system and therefore can be utilized within it to e.g., subsidize protocol operation costs.

When the block-time is large (e.g., 12$s$ as in Ethereum), it is expected that for almost all arbitrage opportunities, all interested players can post their bid within the time interval in an attempt to take advantage of the opportunity. However, when the block-time is small (this is typically the case in layer-2 protocols to increase scalability), perhaps surprisingly, having a connection with lower latency can provide a substantial advantage. This is because when the market update happens close to end of the block time, only players with a faster connection will be able get their transaction included in the block; consequently, they may be able to take advantage of the arbitrage opportunity with a smaller (or even a zero) bid.

Looking ahead, our \sysname policy (which combines both arrival times and bidding) will enable arbitrageurs to prefer bidding even when block times are small, thereby allowing the protocol to capture this value rather than it being lost to external latency infrastructure. 

\mypara{Block or MEV auctions.}
A third widely-used policy auctions off the complete rights to choose and order transactions within a block. Here, the sequencer does not order transactions itself but rather accepts block proposals from external players (often called block \textit{builders}) and chooses the proposal from the builder who pays the most. These auctions initially arose from the realization that significant profit (often referred to as maximal (previously miner) extractable value or MEV~\cite{daian2020flashboys,babel2023clockwork}) can be extracted by manipulating the ordering of user transactions. In the past two years, through companies such as Flashbots and Bloxroute, an MEV marketplace has been created on Ethereum \textit{outside of the protocol} to connect block proposers (entities in charge of proposing or sequencing a block) to block builders (players who find MEV opportunities and order user transactions to take advantage of them)---the result has been the extraction of hundreds of millions of dollars in profit from user transactions~\cite{qin2021bev, wahrstatter2023time}. 

While some MEV (such as arbitrage, which provides incentives for price discovery) is benign and can be done without the knowledge of user transactions, other forms of MEV extraction crucially rely on the transaction data. Recent works~\cite{wahrstatter2023time,yang2022mevsok} have shown such MEV to be significantly detrimental to users. The emergence of such MEV extraction has largely been attributed to the rationality of block proposers as well as the lack of regulation. For example, in traditional financial systems, it is often illegal or at the very least heavily constrained to profit from the knowledge of user transactions (for instance, payment-for-order-flow (PFOF): the selling of user transaction data is illegal in the UK, and, while legal in the US, still requires users to be provided with guarantees of ``best execution'').

A design goal for our work is therefore to design ordering policies that are data-independent, i.e., they do not use transaction data for ordering. This will allow them to be used even when transactions are encrypted at the time of sequencing.

\subsection{Our contributions}
\mypara{\sysname: An ordering policy that combines FCFS and bidding.}
We propose \sysname, an ordering policy that combines both FCFS-style timestamps and first-price auction style bids. Below, we describe several natural goals that went into our design.

\begin{enumerate}
    \item \textbf{Data independence.} The policy should not utilize the transaction data for ordering. This is a natural goal in order to support encrypted transactions and prevent data-dependent MEV attacks on transaction ordering.
    \item \textbf{Low finalization time.} The policy should be able to sequence transactions within a short time $g$ (the specific parameter can be set according to the application). This is important to improve the user experience with the system since transactions will be sequenced within time $g$ after they are received. 
    \item \textbf{Independence of irrelevant transactions.} The ordering between two given transactions should not depend on the presence of other transactions. This is useful to prevent an adversary from inserting irrelevant transactions that results in flipping the ordering between two target transactions. Importantly, this property also ensures that a transaction submitter's strategy need only consider transactions that are relevant to the party's goals---for example if Alice is trying to capture a particular arbitrage opportunity, she need only worry about other transactions affecting that opportunity.
    
    \item \textbf{Inclination to spending via bids instead of latency infrastructure.}
    As mentioned before, investments into latency infrastructure are highly inefficient from the system standpoint since the value spent cannot be utilized effectively by the system. Therefore, a natural goal is to disincentivize latency investment and instead incentivize players to bid for their transactions. Looking ahead, perhaps surprisingly, we find that the pure bidding policy results in a larger latency competition than our \sysname policy which combines bidding with FCFS style timestamps.
\end{enumerate}

\mypara{\sysname details.}
Intuitively, \sysname works by assigning \textit{scores} to transactions based on both their arrival times and their bid. The final ordering is taken to be descending in the transaction scores. More specifically, for a transaction with arrival time $t$ and bid $b$, \sysname assigns it the score $S(\tx) = \pi(b) - t$ where intuitively $\pi$ represents a function for ``buying time''---by increasing the transaction bid, users can reduce their effective timestamp (or equivalently, increase their score). Section~\ref{sec:timeboost} describes how to choose the function $\pi$. 

Importantly, there is a limit to how much time can be ``bought'' through the bid---in particular, no transaction can outbid a transaction received some $g$ time earlier. Such a property is required to ensure the quick finalization of user transactions.  At the same time, transactions received less than $g$ time before can always be outbid; this means that arbitrageurs always have $g$ time to compete for any arbitrage opportunity as opposed to a pure bidding policy and will therefore prefer bidding over latency infrastructure investments. 

We also show that \sysname satisfies all the useful economic properties of first-price all-pay auctions. Further, we show that players spend exactly the same amount in total with \sysname, as they would spend if only latency investment was allowed, except that most of the investment is done through bidding and therefore can be captured within the protocol for e.g., lowering user fees or for protocol development.

%% file: Sections/prelims.tex
\newcommand{\policy}{\mathbb{P}}

\section{Ordering Policies}\label{sec:ordering policies}

\subsection{Preliminaries}
A transaction $\tx$ that arrives at the sequencer can be characterized by a tuple $(\textsf{data}, t, b)$ where $\textsf{data}$ represents the transaction data, $t$ denotes the arrival time, and $b$ denotes the transaction bid (note that when transactions are of different sizes, $b$ can be instead be considered to be a bid per unit size). Let $\txset$ denote the set of all possible transactions; in principle this can be infinite or even uncountable (e.g., if arrival times are in $\mathbb{R}^+$) and our results do hold for these cases. For practical use-cases, typically, arrival times can be assumed to be in $\mathbb{Q}^+$ and bids can be assumed to be in $\mathbb{N}^{\geq 0}$.

An ordering policy now defines how a sequencer orders a finite set $\mathcal{T}'$ of transactions that it has received. A formal definition is given below:

\begin{definition}[(Data-Independent) Ordering Policy]
    An \emph{ordering policy} (or algorithm) $\policy$ takes as input a finite subset $\mathcal{T}' \subseteq \txset$ of transactions and outputs a linear ordering $\policy(\mathcal{T}')$. For 
    $\tx \in \mathcal{T}'$, let $\policy(\mathcal{T}', \tx)$ denote the position of transaction $\tx$ in the ordering $\policy(\mathcal{T}')$. In other words, given $\mathcal{T}'$ and $\tx_a, \tx_b \in \mathcal{T}'$, $\policy$ outputs $\tx_a$ before $\tx_b$ if $\policy(\mathcal{T}', \tx_a) < \policy(\mathcal{T}', \tx_b)$. 
     
    A policy is further called data-independent if it does not make use of the transaction data (i.e., it only uses the arrival time and the bid).
\end{definition}

Since we want our ordering policies to not be based on the transaction content, we only consider data-independent policies for the rest of the paper. For simplicity, we can therefore represent a transaction $\tx$ simply by the tuple $(\tx.t, \tx.b)$. Furthermore, since ties can be broken by some chosen technique, without loss of generality, we can also assume $(\tx.t, \tx.b)$ tuples are unique. While the tie-breaking can be dependent on e.g., transaction ciphertext or metadata, this does not affect our analysis and therefore can be safely ignored for the purpose of our paper.

\subsection{Independence of Irrelevant Transactions (IIT)}

A useful property for our ordering policy to have is to prevent the ordering decision between transactions $\tx_a$ and $\tx_b$ to change depending on what other transactions are being ordered; in other words, the ordering decision should not depend on irrelevant transactions. Intuitively, this is done to ensure that an adversary cannot create dummy transactions in order to flip the ordering decision between two transactions, and so that a party's bidding strategy can ignore transactions irrelevant to that party. We define this property of independence of irrelevant transactions (IIT) below. 

\begin{definition}[Independence of Irrelevant Transactions]
\label{def:iit}
We say that a policy $\policy$ satisfies {\it independence of irrelevant transactions} (IIT) if for any pair of transactions $\tx_a, \tx_b$ and any pair of finite subsets $\mathcal{T}_1,\mathcal{T}_2\subset\mathcal{T}$, the following holds:
\begin{align*}
&&\policy&(\{\tx_a, \tx_b\}\cup \mathcal{T}_1, \tx_a)< \policy(\{\tx_a, \tx_b \}\cup \mathcal{T}_1, \tx_b) \\
&&\Leftrightarrow 
\policy&(\{\tx_a, \tx_b\}\cup \mathcal{T}_2, \tx_a)<\policy(\{\tx_a, \tx_b\}\cup \mathcal{T}_2, \tx_b).
\end{align*}
\end{definition}

\subsection{IIT Implies a Score-Based Policy}

We now show that the IIT property implies that a score-based policy needs to be used---that is, also needs to be independent of the set $\mathcal{T}'$ being ordered. 

Intuitively, a score-based policy works as follows: for transaction $\tx$, it assigns a score $S(\tx)$ based only on the arrival time $\tx.t$ and the bid $\tx.b$. Here too, scoring ties can be broken in a pre-specified manner. The output sequence is then taken to the descending order of transaction scores. Score-based policies are formally defined below:

\begin{definition}[Score-based policy]
A score is a function $S:\mathcal{T}\to \mathbb{R}$ that assigns to each possible transaction $\tx\in\mathcal{T}$ a score $S(\tx)$. An ordering policy $\policy$ is called score-based if there exists a score function $S$ such that $\policy$ sorts transactions according to $S$. In other words, there exists $S$ such that for any $\mathcal{T}' \subseteq \mathcal{T}$ and $\tx_a,\tx_b\in \mathcal{T}'$, it holds that $\policy(\mathcal{T}',\tx_a)<\policy(\mathcal{T}',\tx_b)$ if and only if $S(\tx_a)>S(\tx_b)$.
\label{def:score}
\end{definition} 

For finite $\txset$, we can directly show that IIT implies score-based policies. To show the result for infinite sets, we need to employ the following set-theoretic axiom (defined below) by Cantor~\cite{cantor_axiom}. Similar definitions have also been used in  in the context of utility theory~\cite{debreu}

\begin{property}[Cantor's Axiom~\cite{cantor_axiom}]
We say that a pair $(\policy,\mathcal{T})$ satisfies Cantor's axiom if there exists a {countable} set $\mathcal{T}'\subseteq\mathcal{T}$ such that for any pair of transactions $\tx_a, \tx_b\in\mathcal{T}$ there exists an instance of $\policy$ in which some transaction in $\mathcal{T}'$ is ordered between $\tx_a$ and $\tx_b$. 

Formally there is a finite set $\mathcal{T}''\subset \mathcal{T}$ with $\tx_a, \tx_b \in\mathcal{T}''$ and a $\tx_c \in\mathcal{T}'\cap\mathcal{T}''$ (possibly $\tx_c=\tx_a$ or $\tx_c = \tx_b$) such that

$$\policy(\mathcal{T}'',\tx_a)\leq \policy(\mathcal{T}'',\tx_c)\leq \policy(\mathcal{T}'',\tx_b),$$ 
$$\textrm{or}$$ $$\policy(\mathcal{T}'',\tx_b)\leq \policy(\mathcal{T}'',\tx_c)\leq \policy(\mathcal{T}'',\tx_a).$$
\label{prop:cantor}
\end{property}

We can now establish the following correspondence between IIT and score-based policies.

\begin{theorem}[IIT $\Leftrightarrow$ Score-Based]\label{iit_score}
Let $\txset$ denote the set of all transactions. The following hold for any ordering policy $\policy$:
\begin{enumerate}
\item If $\txset$ is countable, then $\policy$ satisfies IIT if and only if it is score-based.
\item If $\txset$ is uncountable and $(\policy,\txset)$ satisfies Cantor's axiom, then $\policy$ satisfies IIT if and only if it is score-based.
\end{enumerate}
\end{theorem}

\begin{proof}
It is straightforward to see that a score-based algorithm satisfies the independence of irrelevant transactions (since the score of a transaction depends only on itself and not other transactions).

For the opposite direction, we first prove the second part of the theorem (the uncountable case).
We define an order $\prec$ over $\mathcal{T}$ where 
$$\tx_a\prec \tx_b:\Leftrightarrow \policy(\{\tx_a, \tx_b\},\tx_a)<\policy(\{\tx_a, \tx_b\},\tx_b).$$
Since $\policy(\{\tx_a, \tx_b\})$ is a well-defined for any two transactions $\tx_a, \tx_b \in\mathcal{T}$, the order $\prec$ is complete and anti-symmetric. By independence and since  $\policy(\{\tx_a, \tx_b, \tx_c\})$ is a well-defined order for any three transactions $\tx_a, \tx_b, \tx_c \in\mathcal{T}$ we have
\begin{align*}
&\tx_a\prec \tx_b\prec \tx_c\\\Rightarrow &(\policy(\{\tx_a,\tx_b\},\tx_a)<\policy(\{\tx_a,\tx_b\},\tx_b)\text{ and }\policy(\{\tx_b,\tx_c\},\tx_b)<\policy(\{\tx_b,\tx_c\},\tx_c))\\\Rightarrow &\policy(\{\tx_a,\tx_b,\tx_c\},\tx_a)<\policy(\{\tx_a,\tx_b,\tx_c\},\tx_b)<\policy(\{\tx_a,\tx_b,\tx_c\},\tx_c)\\\Rightarrow& \policy(\{\tx_a,\tx_c\},\tx_a)<\policy(\{\tx_a,\tx_c\},\tx_c)\\\Rightarrow &\tx_a\prec \tx_c\end{align*}
Therefore, $\prec$ is transitive. We let $\tx_a\preceq \tx_b$ iff $\tx_a\prec \tx_b$ or $\tx_a=\tx_b$.

 The Cantor axiom and independence imply that there is a countable $\mathcal{T}'\subset \mathcal{T}$ so that the order $\prec$ satisfies that for any $\tx_a,\tx_b\in\mathcal{T}$ there is a $\tx_c\in\mathcal{T}'$ such that
 $$\tx_a\prec \tx_b\Rightarrow \tx_a\preceq \tx_c\preceq \tx_b$$

 By Theorem 1.1 in~\cite{chambers2016}, this, in turn, implies that there is a numerical representation of the order $\prec$ which is a score $S:\mathcal{T}\rightarrow \mathbb{R}$ such that for any two transactions $\tx_a,\tx_b\in\mathcal{T}$ we have $\tx_a\prec \tx_b$ if and only if $S(\tx_a)>S(\tx_b)$.

 For the first part of the theorem, note that the previous argument also works for a countable $\mathcal{T}$ and in that case we can choose  $\mathcal{T}'= \mathcal{T}$ where the Cantor axiom is now trivially satisfied.

\end{proof}

\begin{remark}
The above result extends to the case where the policy creates a weak ordering (which can be made strict through a tie-breaking procedure) rather than a strict ordering of transactions. In that case, Definitions~\ref{def:iit} and~\ref{def:score} are adapted to weak orders, and we get a score that might assign the same value to two different transactions. The relaxation to weak orders is useful for the case that the set of transactions is uncountable and not a subset of the real numbers (e.g.~if $\mathcal{T}=\mathbb{R}^2_+$). In that case, the Cantor axiom is impossible to satisfy for strict orders but satisfiable for weak orders.
\end{remark}

\mypara{Discussion.}
We note that in our context, assuming $\mathcal{T}$ is countable or even finite is safe, as there is a finite smallest time increment for timestamps and a finite smallest bid increment.
Moreover, the ordering policy deals with ordering transactions in a finite time interval and bids will be upper-bounded by the maximum value in the system (e.g., the maximum number of tokens). However, for the subsequent economic analysis, it will be more convenient to work with the continuum where differences in time stamps and bids can be arbitrarily small.

Having proven that score-based algorithms are essentially the only ones satisfying the independence of irrelevant transactions property, we turn to selecting the most natural one among them. Note that FCFS is the scoring function that corresponds to scoring transactions by their timestamp only while scoring transactions only by bids corresponds to the first-price auction solution. In the next section, we show how our scoring policy $\sysname$ corresponds to a simple mixture of these two strategies.

%% file: Sections/timeboost.tex
\section{\sysname Description}
\label{sec:timeboost}
We now formally define the \sysname ordering policy in this section. As mentioned before, we want \sysname to satisfy the independence of irrelevant transactions property (i.e., it needs to be a scoring function based on Theorem~\ref{iit_score}) and also provide low confirmation-latency for transactions. Therefore, we will only allow \sysname\ to consider transactions within a time $g$ interval; this granularity $g$ can be set suitably based on the particular usecase.

\mypara{Basic model.}
Suppose there are $n$ transactions in the $g$ time interval, labeled with $\tx_1,tx_2,\cdots \tx_n$, and sorted by increasing arrival time. Each transaction $\tx_i$ is characterized by a pair of a timestamp or arrival time, denoted by $t_i$, and a bid, denoted by $b_i\geq 0$. Formally, we view a transaction as a tuple of non-negative reals, $\tx_i=(t_i,b_i)\in \mathbb{R}^+\times \mathbb{R}^{\geq 0}$. 

\mypara{\sysname scoring function.} 
Intuitively, for the \sysname scoring function, we propose to allow users to ``buy time'' using their transaction bid; in other words, transactions will be sorted by increasing timestamps (as in FCFS) but now users are allowed to decrease their effective timestamp (i.e., increase their score) through bids.

Formally, the score of a transaction $\tx_i = (t_i, b_i)$ is computed as follows:
\begin{equation}
\hfill S(t_i, b_i) = \pi(b_i) - t_i. \hfill
\end{equation}
where $\pi(b_i)$ denotes the priority or advantage gained by bidding $b_i$. Transactions are now chosen in descending order of their scores.

\mypara{Choosing a bidding function $\pi$.}
To choose the bidding function $\pi$ for \sysname, we start by defining several natural properties that should be satisfied.
\begin{enumerate}
\item $\pi(0)=0$. This normalization implies that paying 0 bid gives no additional advantage.

\item $\pi'(b)>0$ for all $b \in \mathbb{R}^{+}$ where $\pi'$ denotes the first derivative of $\pi$ with respect to the bid. This implies that the priority increases with the bid, which gives incentive to bid more for a higher priority. 

\item $\lim_{b\rightarrow \infty}\pi(b) = g$. This implies that no transaction can outbid a transaction which arrived $g$ time earlier (but any time advantage of less than $g$ can be outbid). Through this, we can guarantee that the transaction ordering can be finalized within time $g$.

\item $\pi''(b)<0$ for all $b \in \mathbb{R}^{+}$ where $\pi''$ denotes the second derivative of $\pi$ with respect to the bid. This means that priority is concave, or equivalently, the cost of producing the (bidding) signal is convex. This is generally necessary to obtain the interior solution of the equilibrium condition.
\end{enumerate}

The simplest bidding function satisfying the above constraints is the function:
\begin{equation}\label{pi_formula}
\hfill \pi(b_i):=\frac{gb_i}{b_i+c} \hfill
\end{equation}
where $c$ is some constant. We will use this as the bidding function for \sysname. In the next section, we provide an economic analysis for \sysname. For this, we will assume that $c = 1$.

\mypara{Complexity.}
For any incoming transaction $\tx = (t, b)$, the sequencer can finalize $\tx$ after a delay of $g-\pi(b_i)$. This is because after this point, no later transaction can outbid $\tx$. If transactions arrive at rate $r$, the space complexity of the sequencing algorithm is $\Theta(r)$ and the computational cost per transaction is $\Theta(\log r)$, assuming pending transactions are stored in a priority queue, ranked by score.

%% file: Sections/economic.tex
\subsection{\sysname Economic Analysis Overview}\label{econ_analysis}

We now describe the model for analyzing the economics for our \sysname ordering policy. The next two sections will describe this analysis in detail.

\mypara{Basic model.} 
Consider an arbitrage opportunity that occurs at some time (w.l.g., this can be taken as time 0). Users (from now on referred to as players) need to now take a decision on (1) how to send their transaction to the sequencer; this corresponds to the investment in latency; and (2) how much extra to bid for their transaction to get higher priority. We will analyze a simple economic model of this decision problem. 

Assume that it costs user $i$ the amount $c_i(t)$ to get its transaction received by the sequencer $t$ time after the arbitrage opportunity arises. The only requirement on $c_i(t)$, for now, is that it is decreasing in increasing $t$. When the arbitrage opportunity arises, a player $i$ has a valuation $v_i$ to have its transaction included for execution the earliest, among those transactions contending for the same opportunity. 

\mypara{Analysis organization.}
We begin with an analysis with two players in Section~\ref{sec:2players}. Within this, we consider different models based on when the latency investment needs to occur. Broadly, we consider two models for latency investment: ex-ante (Section~\ref{subsec:exante}) and ex-post (Section~\ref{subsec:expost}). Ex-ante means that the latency investment needs to happen \textit{before} learning the arbitrage opportunity while ex-post means that the latency investment can occur after learning about the arbitrage opportunity. 

In Section~\ref{sec:nplayers}, we generalize our results to many competing players. 

\section{Analysis of \sysname with 2 Players}
\label{sec:2players}
As a starting point, assume that there are two players with valuations $v_1$ and $v_2$, distributed as per the cumulative distribution functions (CDFs) $F_1$ and $F_2$. That is, the probability that the valuation of player $i$ is less or equal to $x$ is equal to $F_i(x)$. 

For each valuation $v$, the player may choose their specific latency investment.
We can model this as a function $t_i: V \rightarrow \mathbb{R}$, 
such that, $t_i(v)$ is the latency / time chosen by a player $i$ with valuation $v$. For simplicity, assume that the cost functions and value distributions are the same: $c_i(t)=c(t)$ and $F_i =F$. Throughout the paper, we assume that $F:[0,1]\rightarrow[0,1]$ is a uniform distribution with $F_i(x)=x$ iff $x\in [0,1]$, for $i\in \{1,2\}$, when final numerical values are derived. Obtaining numerical values for different distribution functions is very similar to that of uniform distribution, but we choose a uniform for simplicity of exposition. However, most of the computations are done for general distribution functions.

We now consider two different assumptions regarding the investment in latency improvement. In the first model (ex-ante), we assume that the players need to invest in their latency infrastructure in advance: they acquire or rent servers close to the sequencer prior to knowing the value of the arbitrage opportunities they are competing for. In the second (ex-post) model, we assume that the players are able to invest in the latency after they learn about their valuation of the arbitrage. This corresponds to the case where the arbitrage opportunity itself takes some time to be realized\footnote{Example of such an opportunity is a 12-second delay on the Ethereum network for a transaction to be scheduled.}. In this case, the transaction sender can schedule its transaction through the third-party service, which guarantees the delivery of the transaction within some time interval, once the arbitrage opportunity is realized. In both cases, bidding is naturally assumed to be an interim decision, and in fact, one of the biggest advantages, as the valuation is already learned.  

\subsection{Ex-Ante Latency Investment}
\label{subsec:exante}
In this model, players learn their valuations only after they have already invested in latency infrastructure. If players can only compete through latency, the interaction between them becomes a static game. We study equilibria solutions of these games. A similar setting is considered in~\cite{all_pay_contests}. The results obtained in the following two subsections are concrete cases of folk results in the microeconomic theory; however, we include their proofs for completeness.

\subsubsection{Only latency investment}
As a simple first step, we start by analyzing the game where only latency investment is allowed. Let $x_i$ be the amount invested in latency by player $i$  (so that he obtains a delay of $t_i(x_i)$). 
Let $V_i$ denote the valuation random variable of player $i$. 
Then, player $i$ has the following ex-ante payoff: 
\begin{equation*}
\textnormal{Payoff}_i = \begin{cases}
    \Ex[V_i]-x_i & \textnormal{if player invests strictly more than the other player} \\
    \frac{1}{2}\Ex[V_i]-x_i & \textnormal{if he invests an equal amount (assuming random tie-breaking)} \\
    -x_i & \textnormal{otherwise}
\end{cases}
\end{equation*}

First, we note that there is no pure strategy Nash equilibrium solution of the game, in which player strategy sets consist of $\mathbb{R}^+$. It is easy to show this by case distinction on valuations: there are simple deviating strategies in each case. Next, we focus on the mixed equilibrium solution and obtain the following result.

\begin{proposition}
There is a symmetric equilibrium in mixed strategies where each player $i$ chooses $x_i$ uniformly at random on the interval $(0,\Ex[V_i])$.
\end{proposition}

\begin{proof}
By construction, the payoff of player $j$ of playing $x_j\leq E[V_j]$ against the uniform strategy on $(0,E[V_i])$ is
$$F(x_j)E[V_j]-x_j=\frac{x_j}{E[V_i]}E[V_j]-x_j=0.$$
Choosing a strategy $x_j>E[V_j]$ gives a negative pay-off. Therefore, each $0\leq x_j\leq E[V_j]$ is the best response of player $j$, and mixing uniformly among them is also the best response. 

\end{proof}

The above-described equilibrium is unique up to a change of strategy on a null set and in any mixed equilibrium, both players obtain the same payoffs as in this equilibrium. Note that the result is independent of the latency cost function. The only property used is that if a player invests more than the other player in the latency technology, its transaction is scheduled earlier. 

\subsubsection{Budget constraints}
We now model the fact that players may not have access to an arbitrary amount of money they need to invest to improve their latency, but are instead are constrained by a budget. Let $B_i$ denote the budget of player $i$, meaning that player $i$ cannot spend more than $B_i$. We consider an asymmetric case where one (weak) player has a budget $B_1<E[V_i]$ and the other (strong) player has a larger budget with $B_2>B_1$. First, note that similar to the previous section with unlimited access to money, there is no pure strategy Nash equilibrium. Therefore, we switch to mixed strategy equilibrium. Let $F_i$ denote the probability distribution of playing different strategies. 

\begin{proposition}
    There exists a mixed Nash equilibrium solution in the game in which the weak player receives a payoff of $0$ and the strong player receives a payoff of $E[V_i]-B_1$.
\end{proposition}

\begin{proof}
 The following strategy profile in which the first player plays according to the following (mixed) strategy: 

$$
F_1(x)=
\begin{cases}
\frac{x}{E[V_i]}+\frac{E[V_i]-B_1}{E[V_i]},  & x\in (0,B_1],\\
\frac{E[V_i]-B_1}{E[V_i]},\quad  & x = 0,\\
1, & x>B_1,
\end{cases}
$$

the second player plays according to

$$
F_2(x)=
\begin{cases}
\frac{x}{E[V_i]}, \quad & x\in [0,B_1),\\
1, & x\geq B_1,
\end{cases}
$$

is a mixed strategy equilibrium. The first, weak player obtains an expected payoff $0$ for any choice of $0\leq x_1< B_1$. The second, strong player obtains an expected payoff of $E[V_i]-B_1$ for any choice $0< x_2\leq B_1$. Choosing $x_2>B_1$ is wasteful for the second player and will not occur in equilibrium. Thus, both players are indifferent between all pure strategies in support of $F_1$ resp. $F_2$ and for player $2$ choosing an action outside of the support of $F_2$ is dominated. The mixed strategies form a Nash equilibrium.   
\end{proof}

 Similarly to the unconstrained case, the above-described equilibrium is unique up to a change of strategy on a null set and in any mixed equilibrium, both players obtain the same payoffs as in this equilibrium. Also similarly to the previous section, the result is independent of the latency cost function. The only property to derive this result is that if a player invests more than the other player in the latency technology, its transaction is faster (has a lower timestamp).

\subsubsection{Ex-ante Latency with Interim Bidding}
We now analyze the model where both latency and bidding are allowed but the latency is ex-ante. That is, investment in latency happens before players learn their valuations but after learning their valuation players can use bidding to improve the transaction score.

We consider a version where players learn the other players' latency investment decisions before bidding. This models the fact that players will typically play the game repeatedly and can therefore observe latency levels of each other.

In the following let $x=(x_1,x_2)$ be the latency investment levels chosen by the two bidders and let $\Delta:=t_2(x_2)-t_1(x_1)$ be the corresponding difference in latency. W.l.o.g. assume $\Delta\geq0$. First, we consider the case that $\Delta=0$:

\begin{proposition}\label{fullseparation}
    There is a completely separating equilibrium of the bidding game when both bidders have made the same ex-ante investment.
\end{proposition}
\begin{proof}
Given in Section~\ref{subsec:exante-proofs}
\end{proof}

Next, we consider the case that $\Delta\neq 0$.
For the case of different ex-ante investment we get partially separating equilibria where bidders do not bid for low valuations and bid for high valuations. The bidding strategies are asymmetric in general. However, for sufficiently large $g$ the equilibrium becomes approximately symmetric and approximately efficient. See Figure~\ref{bidding} for a graphical illustration.
\begin{figure}
\centering
\includegraphics[width=0.4\textwidth]{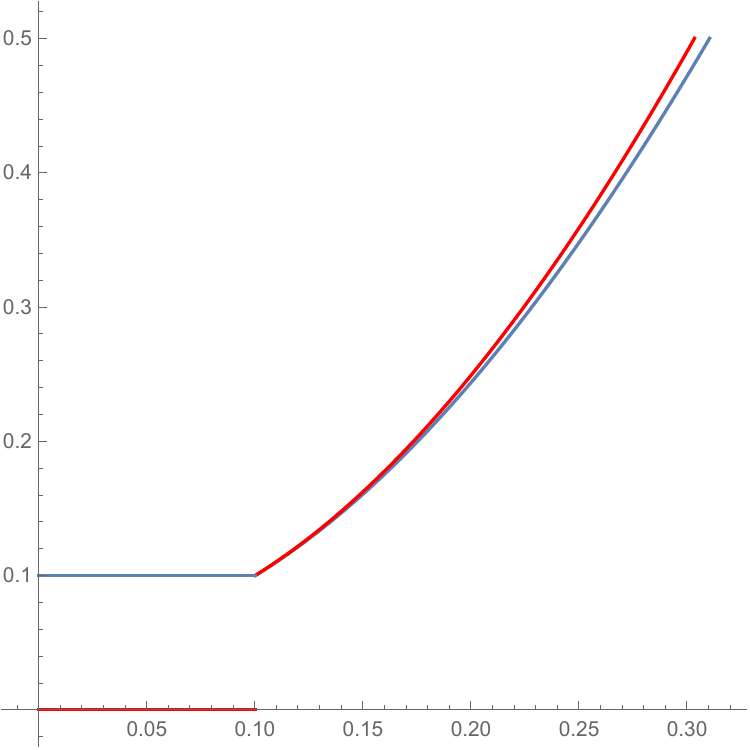}
\caption{Example of equilibrium signaling functions for $g=10$ and $\Delta=0.1$. Timestamps are normalized so that $t_2=0$. The blue function is the equilibrium signal $\pi_1(v)-t_1$ for bidder $1$ as a function of the valuation. The red function is the equilibrium signal $\pi_2(v)-t_2$ for bidder $2$ as a function of the valuation.}
\label{bidding}
\end{figure}
\begin{proposition}\label{partialseparation}
    There is an equilibrium of the bidding game which is separating conditional on bidding: There is a threshold $\sqrt{\frac{\Delta}{g-\Delta}}$, such that a bidder does not bid if his valuation is below the threshold and bids if his valuation is above the threshold. Conditional on bidding, the high latency bidder $i$ produces a higher signal than the low latency bidder $j$ for equal valuations: $\pi_i(v)-t_i>\pi_j(v)-t_j,$ for $v>\sqrt{\frac{\Delta}{g-\Delta}}.$
\end{proposition}
\begin{proof}
Given in Section~\ref{subsec:exante-proofs}
\end{proof}

The equilibrium analysis in Propositions~\ref{fullseparation} and~\ref{partialseparation} indicates how efficient our transaction ordering policy is as a function of the latency investment of bidders.
If bidders have the same latency we have a standard all pay auction which yields a fully efficient outcome.
If there is a difference in latency we have
no bidding for low valuation bidders and approximately 
equal signals produced for equal valuations for high valuation bidders. Conditional on entry, low latency bidders underbid and high latency
bidders overbid relative to the standard all pay strategies.
Efficiency depends on the latency difference and the $g$ parameter.
If $g$ is chosen sufficiently large the auction is approximately efficient. A too low $g$ can be detected by low bidding activity. Hence our transaction policy can strike a balance between fairness, low latency and
efficiency if properly parameterized.

\subsubsection{Proofs}
\label{subsec:exante-proofs}

\begin{proof}[Proof of Proposition~\ref{fullseparation}]

We want to determine bidding signals $\pi_1(v_1,\Delta)$ and $\pi_2(v_2,\Delta)$, which are functions of valuations and the difference in latency. For a given $\Delta$ denote the inverse of $\pi_1(\cdot,\Delta)$ and $\pi_2(\cdot,\Delta)$ by $\tilde{v}_1(\cdot,\Delta)$ and $\tilde{v}_2(\cdot,\Delta)$.
 Then bidder $1$ solves at the interim stage 
$$\max_{\pi\geq0}Pr[\pi-t_1(x_1)\geq \pi_2(v_2,x)-t_2(x_2)]v_1-\frac{\pi}{g-\pi}=F(\tilde{v}_2(\pi+\Delta,\Delta))v_1-\frac{\pi}{g-\pi},$$
We obtain the first order condition:
$$f(\tilde{v}_2(\pi+\Delta,\Delta))v_1\frac{\partial \tilde{v}_2(\pi+\Delta,\Delta)}{\partial \pi}=\frac{g}{(g-\pi)^2}$$

For the uniform distribution, this simplifies to:
$$v_1\frac{\partial \tilde{v}_2(\pi+\Delta,\Delta)}{\partial \pi}=\frac{g}{(g-\pi)^2}.$$
Similarly, for bidder $2$ we obtain 
$$v_2\frac{\partial \tilde{v}_1(\pi-\Delta,\Delta)}{\partial \pi}=\frac{g}{(g-\pi)^2}.$$
The two equations give a system of differential equations that need to be solved for $\pi_1$ and $\pi_2$ or alternatively for $\tilde{v}_1$ and $\tilde{v}_2$.
Alternatively, we can write the system as:
\begin{equation}\label{first_ode}
\hfill\tilde{v}_1(\pi,\Delta)\frac{\partial \tilde{v}_2(\pi+\Delta,\Delta)}{\partial \pi}=\frac{g}{(g-\pi)^2}.\hfill
\end{equation}
\begin{equation}\label{second_ode}
\hfill \tilde{v}_2(\pi,\Delta)\frac{\partial \tilde{v}_1(\pi-\Delta,\Delta)}{\partial \pi}=\frac{g}{(g-\pi)^2}. \hfill 
\end{equation}

The solution to~\eqref{first_ode} and~\eqref{second_ode} in case of equal investment (so that $\Delta=0$) and a symmetric equilibrium is given by the following formula:
\begin{equation}\label{equal_lat_solution}
\hfill    \tilde{v}_1(\pi,0) =  \tilde{v}_2(\pi,0) =\sqrt{2\int_0^{\pi}\frac{g}{(g-\pi)^2}d\pi}=\sqrt{\frac{2\pi}{g-\pi}}. \hfill
\end{equation}

We solve for the signal as a function of the valuation: 
$$
v^2=\frac{2\pi}{g - \pi} \Leftrightarrow \pi = \frac{gv^2}{2+v^2}.  
$$
\end{proof}

\begin{proof}[Proof of Proposition~\ref{partialseparation}]
When $x_1\neq x_2$, we can first sum up~\eqref{first_ode} and~\eqref{second_ode} to obtain a differential equation for the expected payoff $v_1v_2$:
\begin{equation}
\hfill \frac{d (v_1(\pi) v_2(\pi+\Delta))}{d\pi} = \frac{g}{(g-\pi)^2} + \frac{g}{(g-\pi-\Delta)^2}. \hfill
\end{equation}

Integrating both sides of the differential equation above gives the solution: 
\begin{equation}\label{K_equation}
 \hfill   v_1(\pi) v_2(\pi+\Delta)= \frac{\pi}{(g-\pi)}+\frac{\pi+\Delta}{g-\pi-\Delta}+K. \hfill
\end{equation}
To determine the constant we need to determine boundary conditions. For bidder $1$, at the threshold where he is indifferent between bidding and not bidding, we have $\pi_1=0$ and for bidder $2$, at the threshold where he is indifferent between bidding and not bidding, he needs to overcome the handicap, we have $\pi_2=\Delta$. At the threshold bidder $2$ should make the same profit as from pooling, $$v_1(0)v_2(\Delta)=\frac{\Delta}{g-\Delta}\Rightarrow K=0.$$

Combining~\eqref{K_equation} and~\eqref{second_ode} we obtain a separable differential equation: 
\begin{equation}\label{final_generic}
\hfill \frac{d v_1(\pi,\Delta)}{v_1(\pi,\Delta)} = d\pi\frac{g}{(g-\pi-\Delta)^2} \left(\frac{\pi}{g-\pi}+\frac{\pi+\Delta}{g-\pi-\Delta}\right)^{-1}. \hfill
\end{equation}
Combining~\eqref{K_equation} and~\eqref{first_ode} we obtain another separable differential equation: 
\begin{equation}\label{final_genericII}
\hfill \frac{d v_2(\pi+\Delta,\Delta)}{v_2(\pi+\Delta,\Delta)} = d\pi\frac{g}{(g-\pi)^2} \left(\frac{\pi}{g-\pi}+\frac{\pi+\Delta}{g-\pi-\Delta}\right)^{-1}. \hfill
\end{equation}

Integrating both parts of the equation~\eqref{final_generic} solves the (logarithm of) the value as a function of the bid:
$$\ln(v_1(\pi))-\ln(v_1(0))=\int_{0}^{\pi}\frac{g}{(g-\pi-\Delta)^2} \left(\frac{\pi}{g-\pi}+\frac{\pi+\Delta}{g-\pi-\Delta}\right)^{-1}d\pi.$$

Similarly, integrating both parts of the equation~\eqref{final_genericII} solves the (logarithm of) the value as a function of the bid:
$$\ln(v_2(\pi+\Delta))-\ln(v_2(\Delta))=\int_{0}^{\pi}\frac{g}{(g-\pi)^2} \left(\frac{\pi}{g-\pi}+\frac{\pi+\Delta}{g-\pi-\Delta}\right)^{-1}d\pi.$$

To determine the marginal valuations $v_1(0)$ and $v_2(\Delta)$ at which the two bidders start bidding, note that the support of $\pi_i-t_i$ and that of $\pi_j-t_j$ need to coincide for valuations where we have separation of types. Therefore, $v_1(0)=v_2(\Delta)$. Since $v_1(0)v_2(\Delta)=\frac{\Delta}{g-\Delta}$ it follows that $v_1(0)=v_2(\Delta)=\sqrt{\frac{\Delta}{g-\Delta}}$. This is the threshold where bidders start bidding. It follows that for $\Delta\neq0$

$$v_1(\pi)=\sqrt{\frac{\Delta}{g-\Delta}}\exp\left(\int_{0}^{\pi}\frac{g}{(g-\pi-\Delta)^2} \left(\frac{\pi}{g-\pi}+\frac{\pi+\Delta}{g-\pi-\Delta}\right)^{-1}d\pi\right)$$
and
$$v_2(\pi)=\sqrt{\frac{\Delta}{g-\Delta}}\exp\left(\int_{0}^{\pi-\Delta}\frac{g}{(g-\pi)^2} \left(\frac{\pi}{g-\pi}+\frac{\pi+\Delta}{g-\pi-\Delta}\right)^{-1}d\pi\right).$$

To compare the equilibrium signals $\pi_1(v)-t_1$ and $\pi_2(v)-t_2$ for $v>\sqrt{\frac{\Delta}{g-\Delta}}$, we need to compare $\pi_1(v)+\Delta$ to $\pi_2(v)$. 

From the expressions for the valuations as a function of the bid, we can observe (observe that $\frac{g}{(g-\pi-\Delta)^2}\geq \frac{g}{(g-\Delta)^2}$) that 
$$v_1(\pi)>v_2(\pi+\Delta),$$
for $\pi>0$. 
It follows that
$$\pi_1(v)\leq \pi_2(v)-\Delta,$$
for $v\geq\sqrt{\frac{\Delta}{g-\Delta}}$.
\end{proof}

\subsection{Ex-Post Latency with Bidding}

\label{subsec:expost}
We now analyze the ex-post model with bidding; here both the latency investment, and the bid can be made after the valuation is observed. First, we start with only the latency investment decision. 
The expected utility of player $i$ is equal to: 

$$Pr[t(v_i)<t(v_j)]v_i-c(t(v_i)),$$ 

where $j\in\{1,2\}\setminus {i}$.  

We can look at this from a dual perspective: by $v(t)$ we define the inverse of $t(v)$. This is the so-called {\it Revelation Principle}. Instead of some function of the type, we report our type directly. Then, the optimization problem becomes: 
\begin{equation}\label{optimization_2players}
 \hfill    \max_{v} Pr[v\geq v_2]v_1-c(t(v)). \hfill
\end{equation}

By replacing the probability with $F(v)$, we get that it is equivalent to 
\begin{equation*}\label{optimization_2players_simplified}
 \hfill    \max_{v} F(v)v_1-c(t(v)). \hfill
\end{equation*}

By the first order condition, we get:
\begin{equation*}
\hfill v_1f(v)-c'(t(v))t'(v)|_{v=v_1} = 0, \hfill
\end{equation*}

where $f$ is a density function of the valuation distribution $F$. By plugging in $v=v_1$, it is equal to:
\begin{equation}\label{opt_2_simplified}
\hfill v_1f(v_1)-c'(t(v_1))t'(v_1). \hfill
\end{equation} 

For the uniform distribution and cost function $c=\frac{1}{t}$, first order condition  gives the following differential equation: 
\begin{equation}\label{diff_eq_2players_uniform}
\hfill v_1+\frac{t'(v_1)}{t^2(v_1)}=0. \hfill
\end{equation}

Solving this equation gives $t(v)=\frac{2}{c_1+v^2}.$ By the boundary condition that $0$ valuation type should wait infinitely (or equivalently pay $0$ in the latency), we obtain the value of the constant in the solution: $c_1=0$. Therefore, cost incurred is equal to $\frac{1}{t}=\frac{v^2}{2}$. On average each player pays:
\begin{equation*}
\hfill \int_{0}^{1}\frac{v^2}{2}f(v)dv|_{0}^{1} = \frac{1}{6}, \hfill
\end{equation*}

for better latency, before learning their types.
The cost of producing score $s=\frac{gm}{m+1}-t$ is: 
\begin{equation}
\hfill c(s):=m+\frac{1}{t}. \hfill
\end{equation}

We decompose total expenditure into $2$ parts, for bidding and for time, by representing $m$ and $c(t(v))$ as functions of $v$ and taking integrals:

$$b(g):=\int_{0}^{1}m(v)f(v)dv \text{ and } \int_{0}^{1}\frac{1}{t(v)}f(v)dv.$$

\begin{proposition}\label{ex_post_prop}
The limit of $b(g)$ when $g$ tends to infinity is equal to $\frac{1}{6}$. 
$b(g)$ is an increasing function in $g$, for $g$ large enough. 
\end{proposition}

\begin{proof}
    Given in Section~\ref{subsec:expost-proofs}
\end{proof}

The proposition implies that by taking large enough $g$, the system extracts almost all value invested in the latency through bidding.
Starting from some threshold value on $g$, extraction increases with increasing $g$.

We can verify whether the constructed equilibrium is unique by checking the conditions given in~\cite{maskin2003uniqueness}.

\begin{example}
We can calculate a few values of $b(g)$. In particular, $b(1000)\approx 0.1294$, meaning a player pays approximately $77\%$ of the total expenditure in bids, and $b(10000)\approx 0.1537$, meaning a player pays approximately $92\%$ of the total expenditure in bids. 
\end{example}

Note that in the proof of the proposition~\ref{ex_post_prop}, the total investment in both latency and bidding, $c(v)$, is the same value $\frac{v^2}{2}$, as in the case of only investing in the latency. We show that this is not a coincidence. 
In general, assume that there is an arbitrary signaling technology described by an increasing, differentiable cost function $C(s)$.
The following result shows the revenue equivalence of ex-post bidding:

\begin{proposition}\label{rev_equiv_2}
Both players spend the same amount on average for any cost function $C$. 
\end{proposition}
\begin{proof}
    Given in Section~\ref{subsec:expost-proofs}
\end{proof}

The amount spent depends only on the value belief distribution function.

\subsubsection{Proofs}
\label{subsec:expost-proofs}

\begin{proof}[Proof of Proposition~\ref{ex_post_prop}]

The optimization problem of the player in the equilibrium is to minimize cost, subject to the score equation constraint. 
By plugging in $t=\frac{gm}{m+1}-s$, we obtain the minimization problem:
\begin{equation*}
\hfill \min_{m}\left(m+\frac{m+1}{gm-s(m+1)}=:x(m)\right). \hfill
\end{equation*} 

The first order condition on $x(m)$ gives: 
\begin{equation}\label{FOC_2players}
\hfill \frac{dx(m)}{dm} = 1+\frac{gm-s(m+1)-(m+1)(g-s)}{(gm-s(m+1))^2}=1-\frac{g}{(gm-s(m+1))^2}=0, \hfill
\end{equation}
gives that the value of $m$ that minimizes the cost function. The solutions of the last equation are $gm-sm-s=\sqrt{g}$ equivalent to $m=\frac{s+\sqrt{g}}{g-s}$ and $gm-sm-s=-\sqrt{g}$ equivalent to $m=\frac{s-\sqrt{g}}{g-s}$, or the boundary condition $m=0$. For $m=0,$ the value $x(0)=-\frac{1}{s}$, while for $m=\frac{s+\sqrt{g}}{g-s}$, the value 
\begin{equation*}
\hfill x\left(\frac{s+\sqrt{g}}{g-s}\right)=\frac{s+\sqrt{g}}{g-s} + \frac{\frac{s+\sqrt{g}}{g-s} + 1}{g\frac{s+\sqrt{g}}{g-s}-s(\frac{s+\sqrt{g}}{g-s}+1)}=\frac{1 + 2 \sqrt{g} + s}{g - s}. \hfill
\end{equation*}

Accordingly, the marginal cost of producing signal $s$ is:

$$c'(s)=\begin{cases}
\frac{(1 + \sqrt{g})^2}{(g - s)^2},\quad \text{ if } s>-\sqrt{g},\\
\frac{1}{s^2},\quad \text{ if } s\leq-\sqrt{g}.
\end{cases}$$

We solve a similar differential equation as~\eqref{opt_2_simplified}, just with different marginal cost function $c'$, and instead of time function $t$, we have a score function $s$ of valuation $v$. The differential equation becomes: 
\begin{equation}\label{differential_score}
\hfill vf(v)-c'(s)s'(v)=0. \hfill
\end{equation}

We need to solve for the $s(v)$ function. For types $v$ with $\frac{2}{v^2}\geq\sqrt{g}$ who only use latency
we have the same solution as before $$s(v)=-\frac{2}{v^2}.$$

The marginal type who is indifferent between using only latency and using a combination of the two technologies is given by
$$u=\sqrt{\frac{2}{\sqrt{g}}}.$$
which gives the boundary condition
$s(u)=-\sqrt{g}$ for the differential equation describing the behavior of types who choose a signal $s\geq-\sqrt{g}$:

$$v = \frac{(1 + \sqrt{g})^2}{(g - s)^2}s'(v).$$ 
We obtain the solution
\begin{equation}\label{solution_raw}
\hfill s(v) =  (4 c_1 g^{3/2} + 2 c_1 g^2 + 2 c_1 g + g (v^2 - 2) - 4 \sqrt{g} - 2)/(2 c_1 g + 4 c_1 \sqrt{g} + 2 c_1 + v^2). \hfill
\end{equation}

The value of the constant $c$ is obtained from the boundary condition that a zero-value player does not invest and it is equal to 

$$c_1=\frac{1}{(1 + \sqrt{g})^2}.$$

Therefore, plugging in the constant value in the solution~\eqref{solution_raw} and simplifying it gives: 
$$s(v)=\frac{g v^2 - 4 \sqrt{g} - 2}{v^2+2}.$$

Plugging this into the formula of $c(s)$, gives the cost value as a function of valuation $v$:

$$c(v)=\frac{1+2\sqrt{g}+\frac{g v^2 - 4 \sqrt{g} - 2}{v^2+2}}{g-\frac{g v^2 - 4 \sqrt{g} - 2}{v^2+2}}=\frac{v^2}{2}.$$

Separate expenditure in the bidding is calculated by the following formula: 
\begin{align*}
b(g)=&\int_{u}^{1}m(v)f(v)dv = \int_{\sqrt{\frac{2}{\sqrt{g}}}}^1\frac{\frac{gv^2-4\sqrt{g}-2}{v^2+2}+\sqrt{g}}{g-\frac{gv^2-4\sqrt{g}-2}{v^2+2}}dv = \\
&\int_{\sqrt{\frac{2}{\sqrt{g}}}}^{1}\frac{v^2(g+\sqrt{g})-4\sqrt{g}-2}{2g-4\sqrt{g}-2}dv=\\
&\frac{1}{2g+4\sqrt{g}+2}\left(\frac{g+\sqrt{g}}{3}(1-\frac{2}{\sqrt{g}}\sqrt{\frac{2}{\sqrt{g}}})-(4\sqrt{g}+2)(1-\sqrt{\frac{2}{\sqrt{g}}})\right).     
\end{align*}

The dominant term in the nominator above is $g$ and also in the denominator, it is $6g$. Therefore, $\lim_{g\rightarrow \infty}b(g)=\frac{1}{6}$.

\end{proof}

\begin{proof}[Proof of Proposition~\ref{rev_equiv_2}]

We are interested in the equilibrium signaling strategy $s(v)$. Suppose that this strategy is increasing (so no pooling of types) and differentiable. Then, we can define a differentiable function
$$\tilde{C}(v):=C(s(v)).$$
To figure  out what $\tilde{C}(v)$ is, we have to consider an optimization problem with the first player:
$$\max_{v} Pr[v\geq v_2]v_1-C(s(v))=Pr[v\geq v_2]v_1-\tilde{C}(v).$$
Taking first order conditions with respect to $v$ gives:
$$v_1f(v)-\tilde{C}'(v)|_{v=v_1} = 0,$$ that is, $$v_1f(v_1)=\tilde{C}'(v_1).$$ For the uniform distribution: 
$$v_1=\tilde{C}'(v_1). $$
Using the boundary condition $\tilde{C}(0)=0$ and integrating we get
$$\tilde{C}(v_1)=v_1^2/2.$$
More generally:
$$
\tilde{C}(v_1)=\int_{-\infty}^{v_1}vf(v)dv.
$$

\end{proof}

\section{Analysis of \sysname with $n$ players}
\label{sec:nplayers}

In this section, we consider $n$ players with the same valuation distribution as in the $2$ players case. The optimization problem is now the following: 
$$
    \max_{v} Pr[v\geq \max\{v_2,\cdots,v_n\}]v_1-c(t(v)),
$$

similarly to~\eqref{optimization_2players}. 
By replacing the probability with cumulative distribution, this is equivalent to: 
$$
    \max_{v} F_{n-1}(v)v_1-c(t(v)),
$$

where $F_{n-1}(x)$ is a cumulative distribution function of the random variable $X:=\max\{X_1,\cdots, X_{n-1}\}$. By independence we have 
$$F_{n-1}(x)=F(x)^{n-1}.$$

The first-order condition and plugging in $v=v_1$ gives the following differential equation, similar to~\eqref{opt_2_simplified}:
$$
    f_{n-1}(v_1)v_1-c'(t(v_1))t'(v_1) = 0,
$$

where $f_{n-1}(v_1) = (n-1)v_1^{n-2}$ is a density function of maximum among $n-1$ uniformly distributed random variables. The differential equation w.r.t. $t(v)$ becomes: 
$$
(n-1)v_1^{n-1}+\frac{t'(v_1)}{t^2(v_1)}=0. 
$$

Solving the equation gives $t(v)=\frac{n}{c+(n-1)v^{n}}$. The same boundary condition ensures that $c=0$, that is, $t(v)=\frac{n}{(n-1) v^{n}}$. Each player pays: 
$$
\frac{n-1}{n}\int_{0}^{1}v^{n}dv = \frac{n-1}{n}\frac{v^{n+1}}{n+1}|_{0}^{1}=\frac{n-1}{n(n+1)}. 
$$

Together, the players pay $\frac{n-1}{n+1}$, that converges to $1$ as $n$ converges to infinity. 
Note that the first place in the transaction order is given to the maximum-value player. The average valuation of the maximum value player is $\frac{n}{n+1}$, order statistics. This value also converges to $1$ as $n$ tends to infinity.

The analysis is the same as in the case of $2$ players, until the differential equation that solves score function $s$. 
Instead of~\eqref{differential_score}, for $n$ players we solve:
\begin{equation}\label{differential_score_n}
\hfill (n-1)vv^{n-1}-c'(s)s'(v)=0.\hfill
\end{equation}

For types $v$ with $\frac{n}{(n-1)v^2}\geq\sqrt{g},$ who only use latency, we have the same solution as before $$s(v)=-\frac{n}{(n-1)v^2}.$$

Marginal type investing in bidding is: 
$$
u=\sqrt{\frac{n}{(n-1) \sqrt{v}}}.
$$

Plugging in functional forms of $c$ and $s$ in~\eqref{differential_score_n} gives the same limit results as in Proposition~\ref{ex_post_prop}. 
Next, we show a revenue equivalence for $n$ players. The argument is similar to $2$ players' case. 
Assume that there is an arbitrary signaling technology described by an increasing, differentiable cost function $C(s)$.

\begin{proposition}\label{rev_equiv_n}
    All $n$ players spend the same amount on average for any cost function $C$. 
\end{proposition}

\begin{proof}
We are interested in the equilibrium signaling strategy $s(v)$. Suppose that this strategy is increasing (so no pooling of types) and differentiable. Then, we can define a differentiable function
$$\tilde{C}(v):=C(s(v)).$$
To figure  out what $\tilde{C}(v)$ is, we have to consider an optimization problem of the first player:
$$\max_{v} Pr[v\geq \max\{v_2,\cdots,v_n\}]v_1-\tilde{C}(v)=F(v)^{n-1}v_1-\tilde{C}(v).$$
Taking first order conditions with respect to $v$:
\begin{equation*}
\hfill \left[(n-1)v_1f(v)F(v)^{n-2}-\tilde{C}'(v)\right]\big|_{v=v_1}=0, \hfill
\end{equation*}

For the uniform distribution, we get: 
$$(n-1)v_1^{n-1}=\tilde{C}'(v_1). $$
Using the boundary condition $\tilde{C}(0)=0$ and integrating we get
$$\tilde{C}(v_1)=\frac{(n-1)v_1^{n}}{n}.$$
More generally:
$$\tilde{C}(v_1)=\int_{-\infty}^{v_1}(n-1)vf(v)F(v)^{n-2}dv.$$
\end{proof}

%% file: Sections/bidding-comparison.tex
\section{Comparison of \sysname with a Pure Bidding Policy}\label{benchmark}
We now compare \sysname to what to a pure bidding policy. Recall that for the bidding policy, all transactions sent in fixed time intervals of length $g$ are collected, and sorted in decreasing order of their bids. This effectively simulates a first-price all-pay auction for each interval. We note this can be thought of as a quantized version of \sysname, because it produces the same sequence that would be produced by first rounding off each transaction's arrival timestamp to the nearest multiple of $g$ and then applying \sysname.

Generically speaking, a first-price auction where only the winning bidder pays and first-price all-pay auctions are both {\it payoff equivalent} for Bayesian-Nash incentive compatible mechanisms, (see e.g., \cite{rev_payoff_equivalence}).
In our setting, the following result holds for each individual arbitrage opportunity. 
\begin{proposition}[see \cite{rev_payoff_equivalence}]\label{RE_FPA_FPAAP}
The expected payoff of the bidding game where the only the highest bidder pays their bid is equal to the expected payoff in the bidding game where the highest bidder wins but all players pay their bids, independently of valuation distributions. 
\end{proposition}

For simplicity, to compare \sysname with a pure bidding policy, we consider two players. It is straightforward to generalize to more parties. For a given arbitrage opportunity, two cases arise as described below depending on whether transactions can be submitted within the same $g$-time interval as the arbitrage opportunity or not:
\begin{enumerate}
    \item \underline{Both players can submit their transactions within the same $g$ interval.}
    \sloppy For the pure bidding policy, if both players can get their transaction submitted inside the same $g$-time interval as the arbitrage opportunity, then they will both compete for it. It is easy to see that when the valuations of the two parties are the same, the bidding strategy for the pure-bidding policy vs the ex-ante latency with bidding policy will be the same. In other words, in this scenario, \sysname maintains the economic properties of the first-price auction pure-bidding policy.
    
    \item \underline{Only one player can get its transaction within the same $g$ interval.} 
    If only one player can get its transaction inside the same $g$-time interval as the arbitrage opportunity, then in the pure-bidding policy, that player can pay a $0$ bid and still take advantage of it. In contrast, since \sysname does not require discrete boundaries, both players will always have $g$ time to submit their transactions (recall that bidding can be used to get priority over any transaction received up to $g$ time earlier). This means that even for a reasonably small $g$ (say 0.5 sec), both parties will always be able to compete for the opportunity. In equilibrium, this results in bids equal to value of the arbitrage.
\end{enumerate}

\mypara{Analysis for the second case.}
Suppose the first party (denoted by A) can reach the sequencer in $s_1$ time, and the second party (denoted by B) can reach
in $s_2$ time, with $s_1<s_2$. Then, with the pure-bidding policy, A can wait until $g-s_1$ seconds pass since the beginning of a new block creation, 
and send its transaction to the sequencer at exactly $g-s_1$, while B has to send its transaction by time
$g-s_2$ in order to be included in the same block. 

Assuming that arbitrage opportunities are uniformly distributed over the $g$-interval, this means that, with probability $\frac{g-s_1-(g-s_2)}{g} = \frac{s_2-s_1}{g}$, B has no chance to win the race against A, even if it values the arbitrage opportunity much more than A. 
When $g$ is large (e.g., on Ethereum with $12$ sec block times), this latency advantage is not a big issue, as A would only have an advantage with probability $(s_2-s_1)/12$. In contrast, for faster blockchains, or layer-2 rollups which have shorter block-times to achieve scalability, this latency advantage can be significantly more important in the pure-bidding policy vs in \sysname. For instance, when $g = 0.5$sec, A's latency advantage is $24$ times greater than what it was in Ethereum. This means that compared to \sysname, a pure-bidding strategy will either result in substantial latency competition (when $g$ is small) or will not be able to provide low transaction finalization time (since $g$ will be large). 

%% file: Sections/conclusion.tex
\section{Discussion on Sequencer Decentralization}
\newcommand{\pk}{\textsf{pk}}
\newcommand{\sk}{\textsf{sk}}

We now briefly discuss how \sysname can be supported by a decentralized sequencer---i.e., a committee of $\ell$ sequencers (of which at most some $f$ can be dishonest). We only provide possible implementations here; a formal rigorous analysis is outside the scope of this paper. 

In the decentralized setting, transactions to be sequenced are now submitted by users to all sequencers instead of just one. Note that as before, threshold decryption techniques can be used for transaction privacy before ordering. 

The most natural way to support \sysname in a decentralized setting is to have a protocol for sequencers to agree on both the timestamp and the bid of transactions. After this is done, the \sysname policy can simply be applied on the consensus output of the decentralized committee to obtain the final ordering. Agreeing on the bid is easy since we can have the same bid be submitted to all sequencers for a given transaction. Agreeing on the timestamp is a more challenging problem since the same transaction can arrive at different nodes. While it adds significant complexity, one potential technique here is to employ a fair-ordering protocol (this can be as a simple as e.g., computing the median timestamp~\cite{zhang2020oligarchy,kursawe2020wendy} or support more complicated techniques as in~\cite{kelkar2020order, kelkar2021themis, cachin2022quick}). We leave the formal analysis of such a decentralized \sysname implementation to future work.

\section{Conclusion}
We designed \sysname: a policy for transaction ordering that takes into account both transaction arrival times and bids. We showed that any ordering scheme that guarantees the independence of different latency races is a generalized scoring rule. By choosing a suitably designed mixture of timestamps and bids, we showed the economic efficiency of the system: transaction senders spend most of their resources on bidding instead of latency improvement, which can later be used by the protocol for improvement and development. 

\medskip

\mypara{Acknowledgments.}
We are grateful to Lee Bousfield, Chris Buckland, Potuz Heluani, Raul Jordan, Mallesh Pai, Ron Siegel, Terence Tsao as well as participants at the Swiss National Bank Technology and Finance Seminar for interesting discussions and valuable feedback.